\documentclass[10pt, letterpaper]{IEEEtran}

\usepackage{amssymb}
\usepackage{amsfonts}
\usepackage{epsf}
\usepackage{epsfig}
\usepackage{epstopdf}
\usepackage{array}
\usepackage{amsmath}
\usepackage{url}
\usepackage[algoruled, linesnumbered]{algorithm2e}
\usepackage{amsthm}
\usepackage{xifthen}
\usepackage{mathrsfs}
\usepackage[mathscr]{euscript}
\usepackage[sans]{dsfont}
\usepackage[dvipsnames]{xcolor}
\usepackage{graphicx}
\usepackage{caption}
\usepackage{subcaption}

\usepackage{setspace}	

\usepackage{tikz}
\usetikzlibrary{automata,arrows,positioning,calc}

\usepackage{epsf}
\usepackage{epsfig}

\usepackage{array}

\usepackage{times}
\usepackage{graphicx}
\usepackage{amsmath}
\usepackage{amssymb}
\usepackage{amsxtra}
\usepackage{here}
\usepackage{rawfonts}
\usepackage{times}
\usepackage{url}
\usepackage{cite}
\usepackage{amsthm}
\usepackage{graphicx}
\usepackage{caption}

\usepackage{mathtools}
\usepackage{url}
\usepackage[algoruled,linesnumbered]{algorithm2e}

\usepackage{cite}

\usepackage{cite} 

\usepackage[T1]{fontenc}
\usepackage[utf8]{inputenc}
\usepackage{authblk}

\newtheorem{prop}{Proposition}

\newtheorem{remark}{Remark}

\theoremstyle{definition}

\ifCLASSINFOpdf
\else
\fi

\begin{document}

%


\title{\LARGE   \hspace{50 cm} Distributed Uplink Power Control in an Ultra-Dense Millimeter Wave Network: A Mean-field Game Approach }

\author[D. Kopta et al.]
{\normalsize Nof Abuzainab, Walid Saad and Allen B. MacKenzie
	\\ \vspace{-0.4 cm}
	Wireless@VT, Department of Electrical and Computer Engineering, Virginia Tech, Blacksburg, VA, USA, Emails:\{nof, walids,  mackenab\}@vt.edu\\ 
	\vspace{-6ex}
\thanks{This research was supported by the U.S. National Science Foundation under Grant CNS-1526844,}
}

\maketitle


%
\IEEEpeerreviewmaketitle
\vspace{-0.9 cm}


\begin{abstract}
In this paper, a novel mean-field game framework is proposed for uplink power control in an ultra-dense millimeter wave network. The proposed mean-field game considers the time evolution of the mobile users' orientations as well as the energy available in their batteries, under adaptive user association. The objective of each mobile user is then to find the optimal transmission power that maximizes its energy efficiency. The expression of the energy efficiency is analytically derived for the realistic case of a finite size network. Simulation results show that the proposed approach yields gains of  up to $24\%$, in terms of energy efficiency, compared to a baseline in which the nodes transmit according to the path loss compensating power control policy.
\end{abstract}

\vspace{-0.4 cm}
\section{Introduction}

Millimeter wave communications technology is seen as a promising approach to support the high capacity demands of next-generation wireless networks \cite{mmwavecap}. The availability of a large spectrum bandwidth at the millimeter wave frequency bands can potentially be leveraged to meet the the stringent quality-of-service (QoS) requirements of emerging wireless services, particularly in ultra dense in-venue networks such as in stadiums, arenas, shopping malls and transportation hubs.
However, the millimeter wave signal gets severely attenuated if a line-of-sight signal (LOS) does not exist between the base station (BS) and the mobile user (MU). Thus, adaptive resource allocation and user association schemes are necessary to efficiently provide users with their needed service requirements \cite{omid}.

Moreover, in an ultra dense millimeter wave scenario, interference is prominent \cite{slmz}, which necessitates the design of proper interference mitigation schemes \cite{mmwavedens1,mmwavedens2,mmwavedens3} tailored for  millimeter wave networks.
In \cite{mmwavedens1}, a coalition formation game is proposed for user association and bandwidth allocation for an ultra-dense millimeter wave network.
 A centralized approach is  proposed in \cite{mmwavedens2} for the problem of user association and power control in a downlink ultra-dense setting for a millimeter wave network.
A noncooperative game is proposed in \cite{mmwavedens3} for distributed uplink power control of cellular MIMO spatial multiplexing systems.
However, these existing approaches  are either centralized \cite{mmwavedens1} or require significant coordination and control overhead \cite{mmwavedens1,mmwavedens2,mmwavedens3}, which renders them difficult to practically deploy \cite{cognitive}.


One of the promising tools for scalable distributed resource allocation and coordination in a dynamic ultra-dense network settings are mean-field games (MFG). MFG are apropos since they can cope with a massive number of interacting agents in a dynamic environment. Yet, the existing mean-field approaches \cite{meanfielddens1,meanfielddens2} are not specifically designed for millimeter wave networks,  and, consequently, they do not consider the characteristics and challenges of millimeter wave networks.

The main contribution of this paper is a novel MFG-based framework for uplink power control in an ultra-dense millimeter wave network. We formulate the uplink power control as a mean-field game that takes into account the characteristics of millimeter wave networks. In particular, we consider directional beamforming by the base stations and the mobile users. We consider the time evolution of the users' orientations and the energy available in their batteries. Further, we model the randomness of the deployment of BSs as well as MUs using stochastic geometry \cite{mmwavecoverage}. Thus, we consider adaptive user association, in which each user, at each time instant, connects with the BS that provides the required quality-of-service requirement of the MU. We derive the expressions for the user association distributions for a finite size network, as opposed to the prior work that assumes infinite network size. Our results show that the proposed approach can improve energy efficiency by up to $24\%$, compared to a baseline in which the nodes transmit according to a path loss compensating power control policy.

\vspace{-0.4 cm}
\section{System Model}
Consider the uplink scenario in an ultra dense millimeter wave network in which a set $\mathcal{U}$ of  MUs  are transmitting to a set $\mathcal{B}$ of BSs. The locations of the BSs and MUs are distributed according to a uniform Poisson point processes (PPP) of densities $\lambda_b$ and $\lambda_u$, respectively, but are confined to a finite area $\mathcal{A} \subset \mathbb{R}^2$. The area $\mathcal{A}$  is considered to be a ball $\boldsymbol{b}(\boldsymbol{o},R_{\max})$ centered at the origin of radius $R_{\max}$ as done in \cite{slmz}. Since millimeter wave signals are severely attenuated by blockage, the channel quality is highly dependent on whether a LOS signal exists or not. 
Measurements (such as in \cite{mmwavepath}) have shown that different
path loss models exist for the LOS and the non-line-of-sight (NLOS) cases. Thus, $A_L$ and $\alpha_L$ correspond to the path loss coefficient and exponent in case of LOS, and $A_N$ and $\alpha_N$ correspond to the path loss coefficient and exponent in case of NLOS.

In our model, the MUs and the BSs can block the millimeter wave signal. Other blockers, such as trees or other humans which are not participating in communication, are also considered. The locations $\boldsymbol{L_e}=(\boldsymbol{l}_{e,i})_{i \in \mathcal{E}}$ of the  external blockers are distributed according to a uniform Poisson point process  with density $\lambda_e$, where $\mathcal{E}$ is the set of blockers.
All blockers are assumed to have a circular shape with radius $r_B$ as done in \cite{dirppp}.  Then, the blockage probability $p_B(l)$ of any MU $i$'s signal to BS $j$ at distance $l$ is derived as 
$p_B(l)=1-e^{-\lambda(l-r_B)r_B}$
where $\lambda=\lambda_b + \lambda_u + \lambda_e$.

Due to the severe attenuation of millimeter wave signals with distance, all nodes employ directional beamforming. Thus, each millimeter wave BS forms a beam of width $W$ around the direction $(\theta,\beta)$ where $\theta$ and $\beta$  are the spherical elevation and azimuthal angle, respectively.
Hence, according to the sectored antenna model \cite{mmwavecoverage}, each BS's antenna gain at the main lobe is $G_B=\frac{2}{1-cos(\frac{W}{2})}$ and at the side lobe $g_B$.
We assume that the orientations of the BSs are uniformly distributed over $[0,2\pi]$ and are time invariant.
 Each MU, on the other hand, forms a narrow beam of width $w$ around the direction $(\rho,\phi )$ where $\rho$ and $\phi$  are the spherical elevation and azimuthal angle, respectively. Further, due to random changes in the orientation of each MU with time, the azimuthal angle $\phi$ ($\phi \in [0,2\pi]$) of each user evolves as follows \footnote{In in-venue network scenarios, users are characterized by low mobility. Hence, the dynamics of the users' mobility are not considered.}
\begin{equation}
{d \phi}=\mu_\phi dt + \sigma_\phi dW^{u}(t) + d N(t),\label{orientationdyn}
\end{equation}
where $W^{u}(t)$ is a standard Wiener process, $N(t)$ is a process that ensures that the value of $\phi$ remains in the interval $[0,2\pi]$.
The spherical elevation $\rho$ is kept fixed. Thus, the transmit gain $G^{TX}$ from a MU to its serving BS is 
\vspace{-0.2 cm}

\small

\begin{equation}
  G^{\textrm{TX}}(t)=
  \begin{cases}
    G_m, & \text{if } -\frac{W}{2} \leq \phi(t)-\psi_{ub} \leq \frac{W}{2}, \\
   g_m, & \text{otherwise, }  
  \end{cases} \label{GT}
\end{equation}
\normalsize
where $\psi_{ub}$ is the azimuthal angle between the positive x-axis and the direction in which user views BS it is associated to in the horizon plane.
The receive gain at each BS from each MU is 
\vspace{-0.3 cm}
\begin{equation}
  G^{\textrm{RX}}=
  \begin{cases}
    G_B, & \text{if }  -\frac{w}{2} \leq \beta-\psi_{bu} \leq \frac{w}{2},\\
   g_B, & \text{otherwise. } 
  \end{cases}\label{GR}
\end{equation}


Then, from (\ref{GT}) and (\ref{GR}), and since the orientations of the base stations and the mobile users are independent, the total antenna gain $D(t)$ at time $t$ from each user to each BS is a random variable belonging to the set $\mathcal{G}=\{G_B G_m,G_B g_m,g_BG_m, g_Bg_m\}$ whose probability distribution is given by
\vspace{-0.3 cm}

\small 
\begin{equation}
  D(t)=
  \begin{cases}
    G_B G_m, & \text{with probability  } F(t,\psi_{ub})H(t,\psi_{bu}),\\
    G_B g_m, & \text{with probability } (1-F(t,\psi_{ub}))H(t,\psi_{bu}), \\
    g_BG_m, & \text{with probability}F(t,\psi_{ub})(1-H(t,\psi_{bu})),\\
    g_Bg_m, & \text{with probability }(1-F(t,\psi_{ub}))(1- H(t,\psi_{bu})),
 \end{cases}
\end{equation}
where $F(t,\psi_{ub})=Pr(\phi(t)-\frac{W}{2}\leq \psi_{ub} \leq \phi(t)+\frac{W}{2})$ and $H(t,\psi_{bu})=Pr(\phi(t)-\frac{w}{2}\leq \psi_{ub} \leq \phi(t)+\frac{w}{2}).$

\small
\newcounter{tempEquationCounter1} 
\newcounter{thisEquationNumber1}
\newenvironment{floatEq1}
{\setcounter{thisEquationNumber1}{\value{equation}}\addtocounter{equation}{1}
\begin{figure*}[!t]
\normalsize\setcounter{tempEquationCounter1}{\value{equation}}
\setcounter{equation}{\value{thisEquationNumber1}}
}
{\setcounter{equation}{\value{tempEquationCounter1}}
\hrulefill\vspace*{4pt}
\end{figure*}
}
\begin{floatEq1}
\vspace{-0.3 cm}
\begin{equation}
f_L(l,r)=
  \begin{cases}
   \frac{2\pi l \lambda_b (1-p_B(l))e^{-2\lambda_b \pi (\int_{0}^{l}v(1-p_B(v))dv)}}{B_L}, & \hspace{-5 cm}\text{if }  0 \leq r \leq R_{\max} - r_0 \hspace{0.1 cm} \text{or} \hspace{0.1 cm} 0 \leq l \leq R_{\max}-r ,\\
 \frac{ \lambda_bC(l,r)(1-p_B(l))e^{-2\lambda_b \pi (\int_{0}^{R_{\max}-r}v(1-p_B(v))dv)+ \int_{R_{\max}-r}^{l} C(v,r)(1-p_B(v))dv)}}{B_L}, & \text{otherwise. } 
  \end{cases}\label{flos}
\end{equation}
\end{floatEq1}
\normalsize
\vspace{-0.4 cm}
\subsection{User Association}As done in \cite{mmwavecoverage}, each user is associated with the nearest BS, within its communication range $r_0$, such that the gain $AD(t)$ is greater than a required threshold $\eta$ i.e. $AD(t)\geq \eta$, where $A=A_L$ in case of LOS and $A=A_N$, otherwise. Due to the randomness in the BS locations, users' orientations, and the channel between each user and BS, the distance $l$ to the serving BS is a random variable.  In \cite{mmwavecoverage}, the probability density function of the distance $l$ is derived by first finding the probability density functions of the distances to the nearest LOS and NLOS BSs, respectively. However, the network size is assumed to be infinite. In a finite size network, the probability distribution of the distances to the nearest LOS and NLOS BSs depends on the location of the MU, which we obtain in the following proposition.
\vspace{-0.2 cm}
\begin{prop}
Let $r$ be the distance of the MU to the origin, the conditional probability that the distance to the nearest LOS BS is $l$ given that the MU is at distance $r$ from the origin is given by (\ref{flos})
where 

\vspace{-0.3 cm}
 \small\begin{eqnarray}
&&C(l,r)=\theta(l,r) l, \theta(l,r)=\tan^{-1}\Big(\frac{m_1-m_2}{1-m_1m_2}\Big),  \label{Dlr}
\end{eqnarray}
\normalsize
\vspace{-0.4 cm}
\begin{eqnarray}
&&m_1(l,r)=\frac{y_1-r}{x_1}, \hspace{0.2 cm} m_2(l,r)=\frac{y_2-r}{x_2}, \label{slopes}
\end{eqnarray}
\vspace{-0.2 cm}
\small
\begin{eqnarray}
&&\hspace{-1 cm}x_{1,2}(l,r)= \pm \frac{2\delta(l,r)}{r}, \hspace{0.2 cm}y_{1,2}(l,r)= \pm \frac{r}{2}+\frac{R^2_{\max}-l^2}{2r}, \label{coor}
\end{eqnarray}
\normalsize
\scriptsize
\begin{eqnarray}
&&\hspace{-0.8 cm}\delta(l,r)=\frac{\sqrt{(r+R_{\max}+l)(r+R_{\max}-l)(r-R_{\max}+l)(-r+R_{\max}+l)}}{4},\nonumber
\end{eqnarray}
\normalsize

\newcounter{tempEquationCounter2} 
\newcounter{thisEquationNumber2}
\newenvironment{floatEq2}
{\setcounter{thisEquationNumber2}{\value{equation}}\addtocounter{equation}{1}
\begin{figure*}[!t]
\normalsize\setcounter{tempEquationCounter2}{\value{equation}}
\setcounter{equation}{\value{thisEquationNumber2}}
}
{\setcounter{equation}{\value{tempEquationCounter2}}
\hrulefill\vspace*{4pt}
\end{figure*}

}
\begin{floatEq2}
\vspace{-0.5 cm}
\begin{equation}
B_L=
\begin{cases}
1-e^{-2\lambda_b \pi (\int_{0}^{r_0}v(1-p_B(v))dv)}&\text{if }  0 \leq r \leq R_{\max} - r_0 \hspace{0.1 cm} \text{or} \hspace{0.1 cm} 0 \leq l \leq R_{\max}-r ,\\
1-e^{-2\lambda_b \pi (\int_{0}^{R_{\max}-r}v(1-p_B(v))dv)+ \int_{R_{\max}-r}^{r_0} C(v,r)(1-p_B(v))dv)},& \text{otherwise. } 
\end{cases}\label{BL}
\end{equation}
\end{floatEq2}
$B_L$ is the probability that a user has at least one LOS BS and is given by (\ref{BL}).
\end{prop}
\begin{proof}
Let $r$ be the distance from a MU to the origin. If $R_{max}-r_0 \leq r \leq R_{\max}$,
 the probability of the existence of a BS at distance $R_{\max}-r \leq l \leq r_0$ from the MU is given by
$p_{ex}(l)=\frac{C(l,r)}{\pi R^2_{\max}}$, where $C(l,r)$ is the length of the arc of the circle  centered at the MU and of radius $l$ contained in area $\mathcal{A}$. In order to determine $C(l,r)$, we first determine the intersection points of circle and area $\mathcal{A}$. We assume without loss of generality that the coordinates of the MU are given by $(0,r)$. Then, the coordinates of the two intersection points $(x_1,y_1)$ and $(x_2,y_2)$ are given by (\ref{coor}). Next, we determine the slopes of the lines connecting the MU with the intersection points, respectively. The equations can be derived as in (\ref{slopes}).
The equation of the angle between the two lines is obtained as in (\ref{Dlr})
 where
$m_1(l,r)$ and $m_2(l,r)$ are given by (\ref{slopes}).
Then, $C(l,r)=\theta(l,r) l$. Thus, the conditional probability that the distance to the nearest LOS BS is $l$ given that the MU is at distance $r$ from the origin can be derived according to (\ref{flos}).
\end{proof}

\vspace{-0.2 cm}
\begin{prop}
\newcounter{tempEquationCounter} 
\newcounter{thisEquationNumber}
\newenvironment{floatEq}
{\setcounter{thisEquationNumber}{\value{equation}}\addtocounter{equation}{1}
\begin{figure*}[!t]
\normalsize\setcounter{tempEquationCounter}{\value{equation}}
\setcounter{equation}{\value{thisEquationNumber}}
}
{\setcounter{equation}{\value{tempEquationCounter}}
\hrulefill\vspace*{4pt}
\end{figure*}

}
\begin{floatEq}
\vspace{-0.3 cm}
\begin{equation}
f_N(l,r)=
  \begin{cases}
   \frac{2\pi l \lambda_b p_B(l)e^{-2\lambda_b \pi (\int_{0}^{l}vp_B(v)dv)}}{B_N}, & \hspace{-4.5 cm}\text{if }  0 \leq r \leq R_{\max} - r_0,\text{or} \hspace{0.1 cm} 0 \leq l \leq R_{\max}-r ,\\
 \frac{ \lambda_bC(l,r)p_B(l)e^{-2\lambda_b \pi (\int_{0}^{R_{\max}-r}vp_B(v)dv)+ \int_{R_{\max}-r}^{l} C(v,r)p_B(v)dv)}}{B_N}, & \text{otherwise. } 
  \end{cases}\label{fnlos}
\end{equation}
\hspace{-0.5 cm}
\end{floatEq}
\newcounter{tempEquationCounter3} 
\newcounter{thisEquationNumber3}
\newenvironment{floatEq3}
{\setcounter{thisEquationNumber3}{\value{equation}}\addtocounter{equation}{1}
\begin{figure*}[!t]
\normalsize\setcounter{tempEquationCounter3}{\value{equation}}
\setcounter{equation}{\value{thisEquationNumber3}}
}
{\setcounter{equation}{\value{tempEquationCounter3}}
\hrulefill\vspace*{4pt}
\end{figure*}

}
\begin{floatEq3}
\vspace{-0.4 cm}
\begin{equation}
B_N=
\begin{cases}
1-e^{-2\lambda_b \pi (\int_{0}^{r_0}vp_B(v))dv}&\text{if }  0 \leq r \leq R_{\max} - r_0,\text{or} \hspace{0.1 cm} 0 \leq l \leq R_{\max}-r ,\\\
1-e^{-2\lambda_b \pi (\int_{0}^{R_{\max}-r}vp_B(v)dv)+ \int_{R_{\max}-r}^{r_0} C(v,r)p_B(v)dv)}& \text{otherwise. } 
\end{cases}\label{BN}
\end{equation}
\vspace{-0.5 cm}
\end{floatEq3}
Let $r$ be the distance of the MU to the origin, the conditional probability that the distance to the nearest NLOS BS is $l$ given that the MU is at distance $r$ from the origin is given by (\ref{fnlos}).
where $C(l,r)$ is given by (\ref{Dlr}) and $B_N$ is the probability that a user has at least one NLOS BS and is given by (\ref{BN}).

\end{prop}
\vspace{-0.5 cm}

\begin{proof}
The proof follows similar steps as the proof of Proposition 1, and is omitted due to space limitations.
\end{proof}
\normalsize
\vspace{-0.3 cm}
 Let $D_L(t)$ and $d_L$ be the antenna gain, channel coefficient, and distance to the nearest LOS BS. Let $D_N(t)$ and $d_N$ be the antenna gain, channel coefficient and distance to the nearest NLOS BS. Thus, at time $t$, the probability that a user connects to LOS BS given that it is at distance $r$ from the origin is

\vspace{-0.4 cm}
\small
\begin{eqnarray}
&&\hspace{-0.5 cm}\rho_L=B_L\int_{l=0}^{r_0} f_L(l,r)\mathbb{P}(A_L D_L \geq \eta)(\mathbb{P}(d_N \geq l)+\mathbb{P}(A_N D_N\leq \eta))dl\nonumber\\
&&\hspace{-0.5 cm}= B_L\int_{l=0}^{r_0} f_L(l,r) \sum_{g \in \mathcal{G}_L}\mathbb{P}(D_L=g)(p_N(l,r)+\sum_{g \in \mathcal{G}_N}\mathbb{P}(D_N=g))dl\nonumber
\end{eqnarray}
\normalsize
where 
$\mathcal{G}_L=\{g \in \mathcal{G} \hspace{0.1 cm} \text{s.t.} \hspace{0.1 cm} g \geq \frac{\eta}{A_L} \}$, $\mathcal{G}_N=\{g \in \mathcal{G} \hspace{0.1 cm} \text{s.t.} \hspace{0.1 cm} g \leq \frac{\eta}{A_N} \}$
and $\mathbb{P}(d_N\geq l)=p_N(l,r)$, where $p_N(l,r)$ is the probability that no NLOS BS is present at distance less than or equal to $l$ given that the MU is at distance $r$ from the origin and given by (\ref{pN}). Hence, the probability that a user connects to NLOS is: $\rho_N=1-\rho_L$. Then, when the MU is at a distance $r$ from the origin, the conditonal probability distributions  $f^c_L(l,r) $ and $f^c_N(l,r) $ of the distance to the BS given that user $j$ connects to LOS and NLOS BS, respectively are

\vspace{-0.4 cm}
\small
\begin{eqnarray}
&&\hspace{-1 cm}f^c_L(l,r)=f_L(l,r)\mathbb{P}(A_LD_L \geq \eta)(\mathbb{P}(d_N \geq l)+\mathbb{P}(A_ND_N\leq \eta))\nonumber\\
&&\hspace{-0.5 cm}=f_L(l,r) \sum_{g \in \mathcal{G}_L}\mathbb{P}(D_L=g)(p_N(l,r)+\sum_{g \in \mathcal{G}_N}\mathbb{P}(D_N=g)), \hspace{0.1 cm}
\end{eqnarray}
\begin{eqnarray}
&&\hspace{-1 cm}f^c_N(l,r)=f_N(l,r)\mathbb{P}(A_ND_N \geq \eta)(\mathbb{P}(d_L \geq l)+\mathbb{P}(A_LD_L\leq \eta))\nonumber
\end{eqnarray}
\begin{eqnarray}
&&\hspace{-0.5 cm}=f_N(l,r) \sum_{g \in \mathcal{G}'_N}\mathbb{P}(D_N=g)(p_L(l,r)+\sum_{g \in \mathcal{G}'_L}\mathbb{P}(D_L=g)), \hspace{0.1 cm}
\end{eqnarray}
\normalsize
where 
$p_L(l,r)$ is given by (\ref{pL}) and
$\mathcal{G}'_L=\{g \in \mathcal{G} \hspace{0.1 cm} \text{s.t.} \hspace{0.1 cm} g \leq \frac{\eta}{A_L} \}$, $\mathcal{G}'_N=\{g \in \mathcal{G} \hspace{0.1 cm} \text{s.t.} \hspace{0.1 cm} g \geq \frac{\eta}{A_N} \}.$ 
\newcounter{tempEquationCounter4} 
\newcounter{thisEquationNumber4}
\newenvironment{floatEq4}
{\setcounter{thisEquationNumber4}{\value{equation}}\addtocounter{equation}{1}
\begin{figure*}[!t]
\normalsize\setcounter{tempEquationCounter4}{\value{equation}}
\setcounter{equation}{\value{thisEquationNumber4}}
}
{\setcounter{equation}{\value{tempEquationCounter4}}
\hrulefill\vspace*{4pt}
\end{figure*}

}
\begin{floatEq4}
\vspace{-0.3 cm}
\begin{equation}
p_N(l,r)=
\begin{cases}
e^{-2\lambda_b \pi (\int_{0}^{l}vp_B(v))dv}&\hspace{-3 cm}\text{if }  0 \leq r \leq R_{\max} - r_0 \hspace{0.1 cm} \text{or} \hspace{0.1 cm}  r \geq R_{\max}-r_0 \hspace{0.1 cm} \text{and} \hspace{0.1 cm} l \leq R_{\max}-r,\\
e^{-2\lambda_b \pi (\int_{0}^{R_{\max}-r}vp_B(v)dv)+ \int_{R_{\max}-r}^{l} C(v,r)p_B(v)dv)}& \text{otherwise. } 
\end{cases}\label{pN}
\vspace{-0.2 cm}
\end{equation}
\end{floatEq4}

\newcounter{tempEquationCounter5} 
\newcounter{thisEquationNumber5}
\newenvironment{floatEq5}
{\setcounter{thisEquationNumber5}{\value{equation}}\addtocounter{equation}{1}
\begin{figure*}[!t]
\normalsize\setcounter{tempEquationCounter5}{\value{equation}}
\setcounter{equation}{\value{thisEquationNumber5}}
}
{\setcounter{equation}{\value{tempEquationCounter5}}
\hrulefill\vspace*{4pt}
\end{figure*}

}
\begin{floatEq5}
\vspace{-0.3 cm}
\begin{equation}
p_L(l,r)=
\begin{cases}
e^{-2\lambda_b \pi (\int_{0}^{l}v(1-p_B(v))dv}&\hspace{-5 cm}\text{if }  0 \leq r \leq R_{\max} - r_0 \hspace{0.1 cm} \text{or} \hspace{0.1 cm}  r \geq R_{\max}-r_0 \hspace{0.1 cm} \text{and} \hspace{0.1 cm} l \leq R_{\max}-r,\\
e^{-2\lambda_b \pi (\int_{0}^{R_{\max}-r}v(1-p_B(v))dv)+ \int_{R_{\max}-r}^{l} C(v,r)(1-p_B(v))dv)}& \text{otherwise. } 
\end{cases}\label{pL}
\end{equation}
\vspace{-0.6 cm}
\end{floatEq5}

\normalsize
\vspace{-0.3 cm}
\subsection{Utility Function}The received SINR at the associated BS at distance $l$ from the MU's transmission is given by
$\gamma_k(P(t),l)=\frac{A_kl^{-\alpha_k}D(t)P(t)}{N_0B + \bar{M}(t,l)}$
where $k=L$ in case of LOS and $N$, otherwise, $D(t)$ is the antenna gain at time $t$, $\bar{M}(t,l)$ is the aggregate interference at the associated BS, $P(t)$ is the transmission power and $N_0$ is the power spectral density of the additive white Gaussian noise. Since MUs are constrained by their battery, each MU seeks to maximize the energy efficiency function $\xi_k(P(t),l)=\frac{R \cdot q(\gamma_k(P(t),l))}{P(t)}$
%
where $R$ is the transmitter's rate, $q(.)$ is the packet success probability which is a function of the SINR. The average utility from any MU transmission is given by  (\ref{aggrate}). The energy $E$ available at each MU's battery decreases with the transmission power and therefore evolves according to 
\vspace{-0.1 cm}
\begin{equation}
\frac{dE}{dt}=-P(t)dt. \label{Evol}
\end{equation}

\small
\newcounter{tempEquationCounter6} 
\newcounter{thisEquationNumber6}
\newenvironment{floatEq6}
{\setcounter{thisEquationNumber6}{\value{equation}}\addtocounter{equation}{1}
\begin{figure*}[!t]
\normalsize\setcounter{tempEquationCounter6}{\value{equation}}
\setcounter{equation}{\value{thisEquationNumber6}}
}
{\setcounter{equation}{\value{tempEquationCounter6}}
\hrulefill\vspace*{4pt}
\end{figure*}

}
\begin{floatEq6}
\begin{eqnarray}
\bar{U}(P(t))=\int_{r=0}^{R_{\max}} 2\pi r \lambda_u \int_{l=r_B}^{r_0} p_B(l) f^c_N(l,r) \xi_N(P(t),l)+(1-p_B(l))f^c_L(l,r) \xi_L(P(t),l) dl dr. \label{aggrate}
\end{eqnarray}
\vspace{-0.6 cm}
\end{floatEq6}

\normalsize

The objective of each MU is to find the optimal transmission power $P(t) \in [0,P_{\max}]$ for $t  \in [0,T]$ that maximizes its aggregate utility
\vspace{-0.6 cm}

\small
\begin{equation}
J= \int_{t=0}^T \bar{U}(P(t))dt, \label{util}
\end{equation}
\normalsize
 subject to the orientation and energy dynamics constraints  (\ref{orientationdyn}) and (\ref{Evol}). According to  (\ref{orientationdyn}) and (\ref{Evol}), the evolution of the energy and orientation of all users have the same dynamics. Further, all MUs have the same utility function which is given by (\ref{util}) and which depends on the aggregate interference from the remaining MUs according to (\ref{aggrate}). Thus, due to the ultra dense deployment and the homogeneity of the users in terms of the utility and state dynamics, the problem is formulated as a mean-field game as follows.

\vspace{-0.5 cm}

\section{Mean-field Game}

In our problem, the MUs are indistinguishable since they have the same evolution of states as well as the same control and cost functions, as shown in  (\ref{orientationdyn}), (\ref{Evol}), and (\ref{util}). Thus, the MUs satisfy the exchangability property \cite{multiclass} i.e. each MU $k$ needs to only know its state and can implement a homogeneous admissible control $P(t)=\alpha(t,x(t))$. Thus, since the MUs are indistinguishable, we consider the state of a reference MU $\boldsymbol{x}^0(t)=[\phi(t),E(t)]$ that evolves according to (\ref{orientationdyn}) and (\ref{Evol}) with distribution $m(E,\phi,t)$, which is the mean-field. The mean-field $m(\phi,E,t)$ \cite{multiclass} is the solution of the Fokker-Planck-Kolmogorov equation
\begin{eqnarray}
&&\hspace{-2.5 cm}\partial_t m+\mu_\phi \partial_{\phi}m-P\partial_{E} m-\sigma_\phi \partial^2_\phi m=0 \label{FP},
\end{eqnarray}
with initial conditions $m(.,.,.,0)=0$ and reflecting boundary conditions $m(.,.,0,.t)=m(.,.,2\pi,.,t)=0$. In order to derive the aggregate interference on the associated BS for the mean-field case, we use the following properties.
\newcounter{tempEquationCounter7} 
\newcounter{thisEquationNumber7}
\newenvironment{floatEq7}
{\setcounter{thisEquationNumber7}{\value{equation}}\addtocounter{equation}{1}
\begin{figure*}[!t]
\normalsize\setcounter{tempEquationCounter7}{\value{equation}}
\setcounter{equation}{\value{thisEquationNumber7}}
}
{\setcounter{equation}{\value{tempEquationCounter7}}
\hrulefill\vspace*{4pt}
\end{figure*}

}
\begin{floatEq7}
\begin{equation}
p_I(l,r)=
\begin{cases}
2\pi l \lambda_{\mu}dl\hspace{1 cm}&\text{if }  0 \leq r \leq R_{\max} - r_0 \hspace{0.1 cm} \text{or} \hspace{0.1 cm}  r \geq R_{\max}-r_0 \hspace{0.1 cm} \text{and} \hspace{0.1 cm} l \leq R_{\max}-r,\\
D(l,r)\lambda_{\mu}dl&\text{otherwise, } 
\vspace{-0.3 cm}
\end{cases}\label{intd}
\end{equation}
\vspace{-0.5 cm}
\end{floatEq7}
\newcounter{tempEquationCounter8} 
\newcounter{thisEquationNumber8}
\newenvironment{floatEq8}
{\setcounter{thisEquationNumber8}{\value{equation}}\addtocounter{equation}{1}
\begin{figure*}[!t]
\normalsize\setcounter{tempEquationCounter8}{\value{equation}}
\setcounter{equation}{\value{thisEquationNumber8}}
}
{\setcounter{equation}{\value{tempEquationCounter8}}
\hrulefill\vspace*{4pt}
\end{figure*}

}
\begin{floatEq8}
\small
\begin{equation}
Z(r,l,m)=\sum_{g\in\mathcal{G}}\mathbb{P}(D=g)\frac{D}{2\pi}\int_{\theta=0}^{2\pi}\int_{q=rB}^{R_{\max}}p_I(q,d_1(\theta,r,l
))\Big(p_B(q)A_N q^{-\alpha_N}+(1-p_B(q))A_L q^{-\alpha_L}\Big)dq d\theta \label{Zeq}
\end{equation}
\normalsize
\vspace{-0.7 cm}
\end{floatEq8}
\begin{prop}
When the BS is located at a distance $r$ from the origin, the probability
that the  interferer is located in the interval $(l,l+dl)$ is given by (\ref{intd}).
\end{prop}
\vspace{-0.4 cm}
\begin{proof}
Due to the limited communication range of the MUs, the potential interferers are located in the intersection area of the circle centered at the BS with radius $r_0$ and area $\mathcal{A}$. Then, the derivation of the expression in (\ref{intd}) follows a similar argument as in the proof of Proposition 1.
\end{proof}
\vspace{-0.3 cm}
\begin{remark}
When the MU is located at distance $r$ from the origin and the BS at distance $l$ from the MU, the possible distances of the BS from the origin are $d_1(\theta,r,l
)=r^2+l^2-2rlcos\theta$ $\forall \theta \in [0,2\pi]$.
\end{remark}
\vspace{-0.1 cm}

The aggregate interference for a BS at distance $l$ from the MU can be computed as $\bar{M}(t,r,l,m)=\int_{\boldsymbol{x}}m(\boldsymbol{x},t)\alpha(\boldsymbol{x},t)Z(r,l,m)$.
where
$Z(r,l,m)$ is given by (\ref{Zeq}). The expected utility from any MU transmission is given by
\vspace{-0.4 cm}

\small
\begin{eqnarray}
&&\hspace{-0.7 cm}v(P(t),m, \boldsymbol{x})=\int_{r=0}^{R_{\max}} \hspace{-0.4 cm}2\pi r \lambda_u \int_{l=0}^{r_0} p_B(l) f^c_N(l,r) \xi_N(P(t),l,r,m)\nonumber\\
&&\hspace{1.5 cm}+(1-p_B(l))f^c_L(l,r) \xi_L(P(t),l,r,m) dl dr 
\end{eqnarray}
\normalsize
where
\small
$\xi_k(P(t),m,\boldsymbol{x})=\frac{q(\gamma_k(P(t),m, \boldsymbol{x}))}{P(t)},$
$\gamma_k(P(t),l,r,m,\boldsymbol{x})=\frac{A_kl^{-\alpha_k}D(t)}{\sigma^2 + \bar{M}(t,l,r,m)}).$ \label{rate}

\normalsize
The objective of the reference player is to find the optimal power that maximizes the aggregate utility over $[0,T]$ i.e.
\vspace{-0.2 cm}
\begin{equation}
J=\sup_{P(0 \rightarrow T)} \mathbb{E} \Big[\int_{t=0}^T v(P(s),m,\boldsymbol{x})ds + \psi(\boldsymbol{x}(T))\Big]dt.
\end{equation}

In order to find the optimal policy $P(t)$ for a given meanfield $m(t)$, the reference MU uses the following Hamilton Jacobi Bellman (HJB) equation given by
\vspace{-0.3 cm}

\small
\begin{eqnarray}
&&\hspace{-1.2 cm}\partial_t V+\mu_\phi \partial_{\phi}V-P\partial_{E} V+\sigma_\phi \partial^2_\phi V+v(P(t),m(t), \boldsymbol{x}(t))=0, \label{HJB}
\end{eqnarray}
\normalsize
where $V$ is the aggregate utility starting from time $t$ i.e. $V=sup_{P(t \rightarrow T)} \mathbb{E} \Big[\int_{t=s}^T v(P(s),m,\boldsymbol{x})ds + \psi(\boldsymbol{x}(T))\Big].$
Thus, the best response is the one that minimizes the Hamiltonian. The Hamiltonian is given by
$H(\boldsymbol{x},m,\nabla V)=\mu_\phi \partial_{\phi}V-P\partial_{E} V +v(P(t),m(t)), \label{ham}$
where $\nabla$ is the gradient operator. Due to the dependence of the HJB equation of the meanfield $m(t)$, 
the mean-field equilibrium (MFE) is the solution of the HJB equation in (\ref{HJB}) as well as the Fokker-Planck-Kolomogorov equation in (\ref{FP}) .
\vspace{-0.7cm}
\section{Simulation Results}
\vspace{-0.2 cm}
For our simulations, we consider the following values:$\lambda_B=0.08/\text{m}^2$, $\lambda_u=0.03/\text{m}^2$, $\lambda_e=0.01/\text{m}^2$, $\alpha_N=3.88$, $\alpha_L=2.2$, $A_L=A_N=1$, $N_0=-147 \hspace{0.1 cm} \text{dBm/Hz}$, $P_{\max}=0.1 \hspace{0.1 cm} \text{W}$, $u_\phi=\pi/3$, and $\sigma_\phi=\pi/6$. We consider the case in which the MUs have large battery budget. The initial distribution of the users' orientation is chosen to be Gaussian with mean $\frac{\pi}{2}$ and variance $\frac{\pi}{4}$. The considered time duration is $[0,1]$ sec. The utility is computed for both the MFE and for a baseline in which the nodes transmit according to the path loss compensating power control policy.


\begin{figure}[t]
\vspace{-0.2 cm}
\begin{minipage}[t]{0.48\linewidth}
    \includegraphics[width=\linewidth]{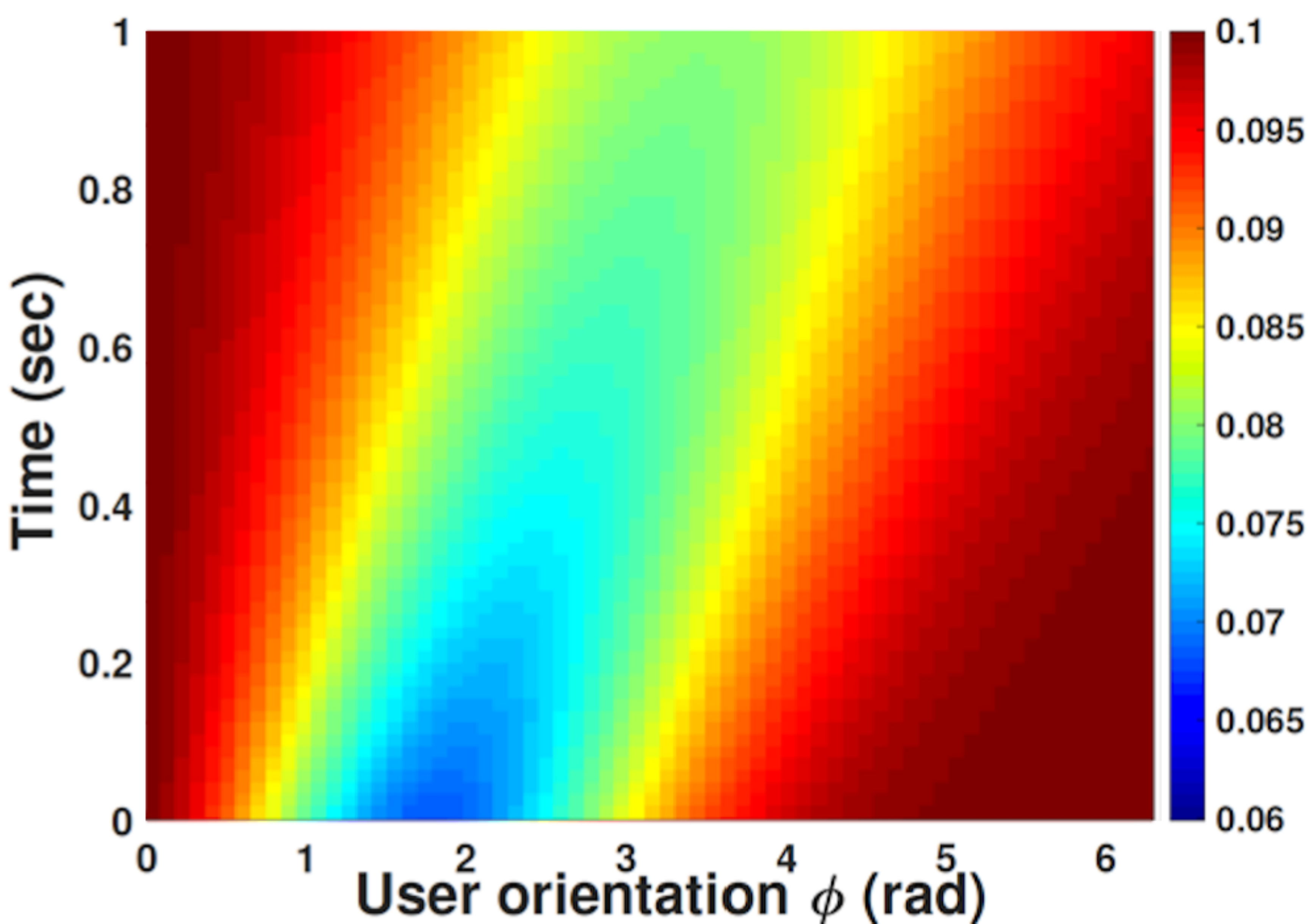}
    \caption{ MFE power versus user's orientation.}
    \label{f1}
\end{minipage}%
    \hfill%
\begin{minipage}[t]{0.48\linewidth}
    \includegraphics[width=\linewidth]{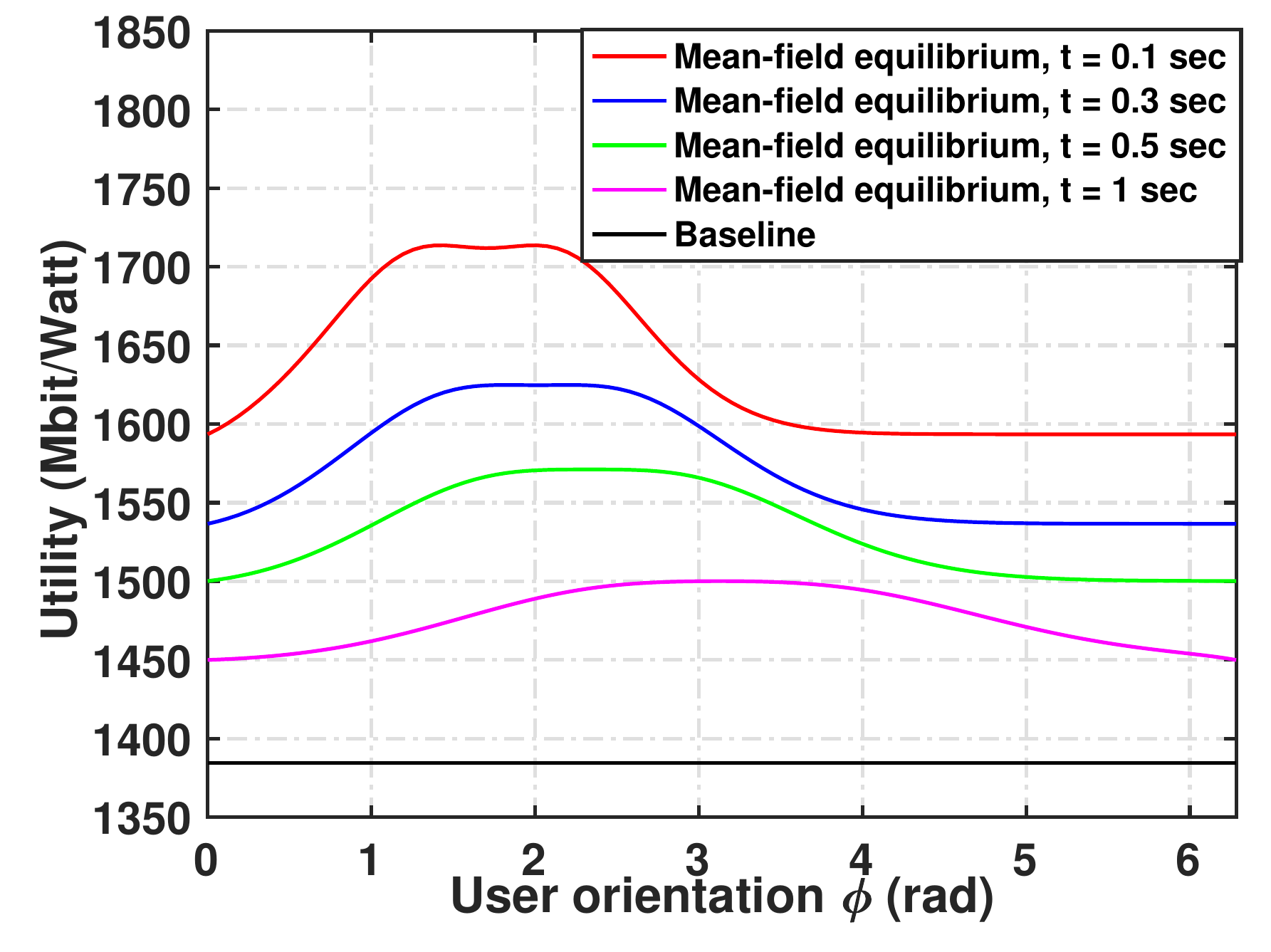}
    \caption{Utility versus user's orientation.}
    \label{f2}
\end{minipage} 
\vspace{-0.8 cm}
\end{figure}

Fig. \ref{f1} shows that, initially, the MFE transmission power decreases as the user orientation approaches $\frac{\pi}{2}$ since the initial distribution is normal with mean $\frac{\pi}{2}$. Thus, the highest proportion of users has an orientation of $\frac{\pi}{2}$, and, thus, users with orientation of $\frac{\pi}{2}$ transmit with the lowest power to reduce interference. As time increases, the distribution of users' orientation flattens, and, thus, the MFE transmission power decreases for a larger range of orientation. Hence, Fig. \ref{f1} shows that the MFE transmission power decreases with the proportion of users' orientation.

Fig. \ref{f2} shows that, at the MFE, the utility decreases with time for all values of the user orientation. This is because, as time increases, the MFE transmission power increases for a larger range of users' orientations, which increases the interference. Fig. \ref{f2} also shows that the  MFE achieves higher utility than the baseline. The increase in the utility when using the MFE reaches up to $24\%$ when $\phi=\frac{\pi}{2}$ and $t=0.1$ seconds.

\vspace{-0.3 cm}
\section{Conclusion}

In this paper, we have proposed a mean-field game  to solve the power control problem in ultra-dense millimeter wave network. Our results have quantified the performance gains achieved when using the mean-field approach compared to the baseline in which the nodes transmit with maximum power. Future work will consider applying our framework to more mobile scenarios.

\vspace{-0.2 cm}
\def\baselinestretch{0.75}

\end{document}